\documentclass[onecolumn,prx,longbibliography,nofootinbib]{revtex4-2}

\usepackage[dvips]{graphicx} 
\usepackage{amsfonts,amscd,amsmath,amsthm}
\usepackage{enumerate}
\usepackage{epsfig}
\usepackage{subfigure}
\usepackage{xcolor}
\usepackage[colorlinks = true]{hyperref}
\usepackage{physics}
\usepackage{epstopdf}
\usepackage{framed}
\usepackage{multirow}
\usepackage{color}
\usepackage{comment}
\usepackage[ruled,vlined]{algorithm2e}
\usepackage[most]{tcolorbox}

\graphicspath{{./figure/}}

\usepackage{tikz}
\usetikzlibrary{tikzmark, calc, fit, positioning}
\usetikzlibrary{shapes}
\usetikzlibrary{quantikz}
\usepackage{pgfplots}
\pgfplotsset{compat=newest}

\newtheorem{theorem}{Theorem}
\newtheorem{lemma}{Lemma}

\newtheorem{question}{Question}

\newtcolorbox[auto counter]{mybox}[2][]{
	enhanced,
	breakable,
	colback=blue!5!white,
	colframe=blue!75!black,
	fonttitle=\bfseries,
	title=Box \thetcbcounter: #2,#1
}

\begin{document}

\title{Enhanced Analysis for the Decoy-State Method}

\author{Zitai Xu}
\affiliation{Center for Quantum Information, Institute for Interdisciplinary Information Sciences, Tsinghua University, Beijing 100084, China}
\author{Yizhi Huang}
\affiliation{Center for Quantum Information, Institute for Interdisciplinary Information Sciences, Tsinghua University, Beijing 100084, China}
\author{Xiongfeng Ma}
\email{xma@tsinghua.edu.cn}
\affiliation{Center for Quantum Information, Institute for Interdisciplinary Information Sciences, Tsinghua University, Beijing 100084, China}

\begin{abstract}
Quantum key distribution is a cornerstone of quantum cryptography, enabling secure communication through the principles of quantum mechanics. In reality, most practical implementations rely on the decoy-state method to ensure security against photon-number-splitting attacks. A significant challenge in realistic quantum cryptosystems arises from statistical fluctuations with finite data sizes, which complicate the key-rate estimation due to the nonlinear dependence on the phase error rate. In this study, we first revisit and improve the key rate bound for the decoy-state method. We then propose an enhanced framework for statistical fluctuation analysis. By employing our fluctuation analysis on the improved bound, we demonstrate enhancement in key generation rates through numerical simulations with typical experimental parameters. Furthermore, our approach to fluctuation analysis is not only applicable in quantum cryptography but can also be adapted to other quantum information processing tasks, particularly when the objective and experimental variables exhibit a linear relationship.
\end{abstract}

\maketitle 
\tableofcontents
\clearpage

\section{Introduction} \label{Sec:Intro}
Quantum key distribution (QKD) is among the most practical technologies in quantum information science, allowing distant parties to establish secure keys based on the fundamental principles of quantum mechanics~\cite{bennett1984quantum, ekert1991Quantum}. With its foundation in information-theoretic security~\cite{lo1999Unconditional}, QKD ensures that any eavesdropping attempt is detectable, offering unprecedented security for applications such as distributed quantum computing and secure cloud services. The practicality of QKD has been validated through various large-scale global deployments of quantum networks~\cite{peev2009secoqc, sasaki2011field, Tang2016MDInet, Chen2021implementation, chen2021integrated}, which highlight the potential and feasibility of this security technology.

Since the original BB84 protocol~\cite{bennett1984quantum}, significant progress has been made in both theoretical and experimental aspects of QKD~\cite{Xu2020Secure}. These efforts have led to the development of sophisticated protocols such as measurement-device-independent QKD~\cite{lo2012Measurement}, twin-field QKD~\cite{lucamarini2018overcoming}, and mode-pairing QKD~\cite{zeng2022mode}, each pushing the boundaries of secure quantum communication.

Despite differences in encoding schemes and key extraction mechanisms, most practical QKD implementations rely on the decoy-state method~\cite{hwang2003decoy, Lo2005Decoy, wang2005decoy}. This method tackles the challenge of producing high-quality single-photon sources by providing a secure way to estimate key rates, even in the presence of multi-photon emissions that are vulnerable to photon-number-splitting attacks~\cite{Brassard2000}. By mitigating this vulnerability, the decoy-state method ensures security and improves key generation rates in real-world scenarios, where practical photon sources emit both single-photon and multi-photon states.

A key innovation of the decoy-state method lies in its ability to leverage statistical data from different quantum states to estimate the single-photon component and its associated phase error rate. By incorporating the Gottesman-Lo-Lütkenhaus-Preskill security analysis~\cite{gottesman2004security}, a secure key rate can be derived. However, in practical applications, statistical fluctuations due to finite data sizes significantly affect the final key rate~\cite{renner2008security}. These fluctuations pose a major challenge, especially given the large number of variables involved in the decoy-state method. Furthermore, the nonlinear dependence of the key rate on the phase error rate adds complexity to the analysis. Addressing these fluctuation issues is crucial to enhancing key generation performance and reducing the number of experimental rounds, thus making QKD protocols more practical, particularly for long-distance communication. 

While previous studies~\cite{Ma2005practical,marco2012tight,hayashi2012concise, Lim2014, Zhang2017improved} have made progress in addressing these challenges, primarily by separately estimating bounds for the single-photon component and phase error rate, the nonlinear nature of the Shannon entropy function in the key rate formula still poses a significant limitation. This suggests that further refinements in the finite-key analysis of the decoy-state method are necessary.

In this work, we tackle the challenge of statistical fluctuations by introducing a linear relaxation of the nonlinear dependence of the key rate on the single-photon phase error rate. Although this linearization might slightly reduce the final key rate, we show that, with careful selection of the slope and intercept, the impact is minimal. Moreover, due to the convex nature of the Shannon entropy function, an appropriately chosen tangent line can provide an effective linear bound. This approach offers two key benefits: it simplifies the theoretical analysis, allowing us to derive tighter key rate bounds in the absence of statistical fluctuations, and it facilitates more precise and efficient fluctuation analysis. Additionally, through this linear relaxation, we identified the optimal eavesdropping strategy and incorporated it to further refine our bounds.

Building on this, we present a joint statistical fluctuation analysis framework that is applicable when variables and objectives exhibit linear relationships. Specifically, in the context of the decoy-state method, we extend the previous approach~\cite{Zhang2017improved} by conducting the analysis after key-sifting, where the photon-number component statistics are fixed but unknown. By modeling the distribution of valid detection events across different decoy states using Bernoulli variables, we integrate this into our fluctuation analysis framework. Leveraging the linear bounds, we apply the Chernoff bound to achieve tighter and more concise results.

The remainder of this paper is organized as follows. In Sec.~\ref{Sec:Preliminary}, we review the decoy-state method and the Chernoff bound. In Sec.~\ref{Sec:Linear}, we analyze the eavesdropper's optimal attacking strategy and derive a linear lower bound for the key rate in the absence of statistical fluctuations. In Sec.~\ref{Sec:Joint}, we present a fluctuation analysis framework by jointly evaluating the variables using the Chernoff bound, followed by combining the linear key-rate bound to obtain the final key rate. All relevant formulas for experimental data post-processing are provided in subsection~\ref{subsec:exp}. In Sec.~\ref{Sec:Simulation}, we present simulation results, and finally, we conclude in Sec.~\ref{sc:conclusion} with a discussion and outlook.

\section{Preliminary} \label{Sec:Preliminary}

\subsection{Decoy-state method} \label{sub:decoy}
The core idea of the decoy-state method is that Alice and Bob encode their quantum information into optical pulses with different intensities, enabling them to monitor the channel performance on different photon numbers in real time. For simplicity, we focus primarily on the one-decoy state scheme \cite{Ma2005practical}, as presented in Box \ref{box:decoy}. The methodology can be extended to other schemes involving multiple decoy states in a similar manner.

\begin{mybox}[label={box:decoy}]{One-decoy-state method}
	\begin{enumerate}[(1)]
		\item 
		State preparation: In the $i$-th round, Alice prepares a coherent-state pulse $\ket{\alpha_i} = \ket{\sqrt{\mu_i}e^{\mathrm{i} \theta_i}}$, where $\mu_i \in \{\mu, \nu\}$ represents the intensity selected with probabilities $p_\mu$ and $p_\nu$, respectively, and $\theta_i \in [0,2\pi)$ is the random phase. Alice encodes it into four BB84 quantum states, $\ket{0},\ket{1},\ket{+},\ket{-}$, for instance, using polarization.
				
		\item
		Measurement: Alice sends the optical pulse to Bob, who measures it randomly in either the $X$ or $Z$ basis.
		
		\item 
		They repeat the above steps for $N$ rounds. Alice sends out $N_\mu, N_\nu$ rounds for the intensities $\mu,\nu$, respectively, where $N_\mu\approx p_\mu N$, $N_\nu\approx p_\nu N$, and $N_\mu+N_\nu=N$. Correspondingly, Bob receives $n_\mu$ and $n_\nu$ valid detection clicks for the intensities $\mu$ and $\nu$, respectively.
					
		\item 
		Announcement: Bob announces the rounds where he obtains valid clicks, along with his basis choices. Alice then announces her selected intensities ($\mu$ or $\nu$) and bases. They also announce raw bits corresponding to the decoy states, where $\mu_i \neq \mu$. 								
				
		\item 
		Key mapping: They keep the bits from the intensity $\mu$ when they choose the same basis as the raw key.
		
		\item 
		Parameter estimation: Alice and Bob use the decoy-state method to estimate the fraction of the raw key derived from single-photon states and calculate the associated error rate.		

		\item 
		Post-processing: Alice and Bob reconcile their raw key via a public classical channel, followed by privacy amplification to generate the final secure key.
	\end{enumerate}
\end{mybox}

Several notes regarding the scheme are as follows. First, the approximate equations in Box~\ref{box:decoy} stem from statistical fluctuations, which become exact as $N \rightarrow \infty$. Second, while we assume Alice and Bob announce all information from the decoy states for simplicity, they could also generate secure key bits from decoy states. Third, Alice and Bob can first perform information reconciliation to determine the exact bit error rate before estimating the phase error rate~\cite{Fung2010Finite}, allowing them to precisely count bit errors without statistical fluctuations affecting the results. Lastly, they can use pre-shared secure keys to encrypt communication during the information reconciliation process and then ``return" those secure keys after the final key generation, ensuring the net key rate remains unchanged~\cite{du2024advantage}.

Here, we assume that Alice and Bob use phase-randomized coherent state sources~\cite{Lo2005Decoy}, which can be expressed as
\begin{equation} \label{eq:photonnumber}
	\frac{1}{2\pi}\int_{0}^{2\pi} \ketbra{\sqrt{\mu}e^{\mathrm{i} \theta}}d\theta = \sum_{k=0}^{\infty} \frac{\mu^k e^{-\mu}}{k!} \ketbra{k},
\end{equation}
where the probabilities of different Fock states follow a Poisson distribution. In this setting, we can think of there being distinct channels between Alice and Bob, each corresponding to different photon numbers, labeled by $k$. In principle, Alice and Bob can independently analyze detection probabilities and error rates for each of these channels. Note that this analysis can be modified for other types of photon sources, such as thermal states~\cite{Ma2008PhD}, to suit different experimental setups. In practice, phase randomization can be approximated effectively using discrete phases~\cite{Cao2015discrete}.

Due to the risk of photon-number-splitting (PNS) attacks, the multi-photon components ($k \ge 2$) do not contribute to a positive key rate in the standard security analysis~\cite{gottesman2004security}. Define the gain of a $k$-photon component, $Q_k$, as the probability that Alice sends out a $k$-photon state and Bob successfully obtains a valid detection click. The corresponding error rate for the $k$-photon component is denoted by $e_k$. The yield of the $k$-photon state, $Y_k$, is defined as the conditional probability that Bob obtains a valid click given that Alice sends out a $k$-photon state. Given that the source follows a Poisson distribution, we have:
\begin{equation}
	Q_k = \frac{\mu^k e^{-\mu}}{k!} Y_k.
\end{equation}

Define the total gain, $Q_{\mu}$, for a signal state as the probability that Alice sends out a pulse of intensity $\mu$ and Bob obtains a valid click. The corresponding total error rate is denoted by $E_\mu$. Therefore, the following relations hold:
\begin{equation} \label{eq:QEeqs}
	\begin{split}
		Q_\mu &= \sum_{k=0}^{\infty} \frac{\mu^k e^{-\mu}}{k!} Y_k, \\
		Q_\mu E_\mu &= \sum_{k=0}^{\infty} \frac{\mu^k e^{-\mu}}{k!} Y_k e_k.
	\end{split}
\end{equation}
Here, $\mu$ is a known value chosen by Alice, while $Y_k$ and $e_k$ are unknown and potentially controlled by an eavesdropper, Eve. Thus, it is crucial for Alice and Bob to monitor $Y_k$ and $e_k$ in real time to ensure the security of key distribution.

\subsection{Key Rate}
In QKD, the key rate can be defined in different ways. One definition, denoted by $r$, represents the ratio of the final secure key bits to the raw key bits. Specifically, if Alice and Bob generate $n_{\mu}$ raw bits in the signal states when they select the same basis, the final number of secure key bits should be $r n_{\mu}$. In an error correction scheme without encryption, the privacy amplification process should use a hashing matrix of dimensions $r n_{\mu} \times n_{\mu}$. The key rate $r$ is given by~\cite{Lo2005Decoy, gottesman2004security}:
\begin{equation}\label{Eq:KeyRate_small_r}
	r = \frac{Q_1}{Q_{\mu}}\left[1-h\left(e_1\right)\right] - I_{ec},
\end{equation}
where $I_{ec} = f h(E_{\mu})$ represents the cost for information reconciliation, $h(x) = -x\log(x) - (1-x)\log(1-x)$ is the binary entropy function, and $f$ is the error correction efficiency.

Another definition, denoted by $R$, represents the ratio of secure key bits to the total number of laser pulses sent, also known as the key rate per emitted pulse. In this case, the key rate formula is:
\begin{equation}\label{Eq:KeyRate}
	\begin{split}
		R &= p_{\mu}\left\{ Q_1[1-h(e_1)] - I_{ec} Q_{\mu} \right\} \\
		&= p_{\mu} r Q_{\mu},
	\end{split}
\end{equation}
where $p_{\mu}$ is the probability that Alice chooses the signal state. This key rate means that if Alice sends a total of $N$ pulses --- including both signal and decoy states --- the final secure key will contain $RN$ bits. Since the total number of pulses $N$ reflects the running time of the experiment, $R$ is more commonly used for comparing the efficiency of different QKD protocols.

In an experiment, $Q_\mu$ can be directly obtained, and $I_{ec}$ can be determined after information reconciliation. The central idea of the decoy-state method is to estimate $Q_1$ and $e_1$ based on statistics from different signal and decoy states. Since $Q_1 = \mu e^{-\mu} Y_1$, our main focus is on bounding the privacy amplification term,
\begin{equation}
	Y=Y_1\left[1-h\left(e_1\right)\right],
\end{equation}
using all the information available to Alice and Bob.

Note that, in general, $Y_k$ may differ for the $X$ and $Z$ bases when $k \geq 3$, as Eve can potentially obtain information about the basis for multi-photon components. Therefore, in a more general analysis, we would need to define two distinct sets of variables: $\{Y_k^X, e_k^X\}$ and $\{Y_k^Z, e_k^Z\}$. However, for simplicity, we consider a symmetric QKD scheme in which the probabilities of selecting the $X$ and $Z$ bases are equal. We define $Y_k$ and $e_k$ as the yields and bit error rates averaged across both the $X$ and $Z$ bases. By doing so, the overall key rate will be lower than it would be if we analyze the $X$ and $Z$ bases separately, which can be seen from the following inequality:
\begin{equation}
	\begin{split}
		\hat{Y}_1^X[1 - h(\hat{e}_1^Z)] + \hat{Y}_1^Z[1 - h(\hat{e}_1^X)] &\geq (\hat{Y}_1^X + \hat{Y}_1^Z) \left[ 1 - h \left( \frac{\hat{Y}_1^X \hat{e}_1^Z + \hat{Y}_1^Z \hat{e}_1^X}{\hat{Y}_1^X + \hat{Y}_1^Z} \right) \right] \\
		&= \hat{Y}_1[1 - h(\hat{e}_1)],
	\end{split}
\end{equation}
where the hat symbol in $\hat{Y}$ and $\hat{e}$ indicates that these are estimators rather than true values. The inequality holds because the binary entropy function $h(x)$ is convex.

\subsection{Chernoff Bound}\label{subsec:Chernoff}
In the analysis of the decoy-state method, we often need to work with probabilities, whereas in practical experiments, we only have access to observed frequencies. Statistical fluctuations can lead to deviations between the underlying probabilities and the observed frequencies. To account for such deviations, we employ the Chernoff bound~\cite{curty2014finite, Lim2014, Zhang2017improved}, which provides a way to quantify the impact of statistical fluctuations.

\begin{lemma}[Chernoff Bound]\label{lemma:chernoff}
	Suppose $X_1, \dots, X_n$ are independent discrete random variables taking values in $[0, 1]$. Let $X$ denote the sum of these random variables. Then $\forall \delta > 0$,
	\begin{equation}\label{eq:chernoff_upper}
		\Pr(X \geq (1 + \delta)\mathbb{E}[X]) \leq \left(\frac{e^{\delta}}{(1 + \delta)^{1 + \delta}}\right)^{\mathbb{E}[X]},
	\end{equation}
	and for $0 < \delta < 1$, 
	\begin{equation}\label{eq:chernoff_lower}
		\Pr(X \leq (1 - \delta)\mathbb{E}[X]) \leq \left(\frac{e^{-\delta}}{(1 - \delta)^{1 - \delta}}\right)^{\mathbb{E}[X]}.
	\end{equation}
\end{lemma}

The expressions on the right-hand side of Eqs.~\eqref{eq:chernoff_upper} and \eqref{eq:chernoff_lower} can be cumbersome to work with directly. To simplify the process of bounding these probabilities, we use the following relaxation.

\begin{lemma}\label{lemma:simpler-chernoff}
	For any $\delta > 0$, we have
	\begin{equation}
		\frac{e^{\delta}}{(1 + \delta)^{1 + \delta}} < e^{-\frac{\delta^2}{2 + \delta}},
	\end{equation}
	which implies
	\begin{equation}
		\Pr(X \geq (1 + \delta)\mathbb{E}[X]) \leq e^{-\frac{\delta^2}{2 + \delta}\mathbb{E}[X]}.
	\end{equation}
	Similarly, for any $0 < \delta < 1$, we have
	\begin{equation}
		\frac{e^{-\delta}}{(1 - \delta)^{1 - \delta}} \leq e^{-\frac{\delta^2}{2 + \delta}},
	\end{equation}
	which implies
	\begin{equation}
		\Pr(X \leq (1 - \delta)\mathbb{E}[X]) \leq e^{-\frac{\delta^2}{2 + \delta}\mathbb{E}[X]}.
	\end{equation}
\end{lemma}

With this lemma, the upper and lower bounds for $X$ become symmetric, allowing us to determine a single value of $\delta$ that bounds the probability on both sides, given a small security parameter $\varepsilon$. We can formulate this question as follows.

\begin{question}\label{ques:lowerBound}
	Suppose we have a random variable $X = \sum_{i=1}^{n} X_i$, where each $X_i$ takes values in $[0, 1]$. Given a security parameter $\varepsilon$, find the smallest $\delta$ such that 
	\begin{equation}
		\Pr\left(\frac{X}{1 + \delta} \geq \mathbb{E}[X]\right) \leq \varepsilon,
	\end{equation}
	and
	\begin{equation}
		\Pr\left(\frac{X}{1 - \delta} \leq \mathbb{E}[X]\right) \leq \varepsilon.
	\end{equation}
\end{question}

If the expectation $\mathbb{E}[X]$ is known, solving the equation $e^{-\frac{\delta^2}{2 + \delta} \mathbb{E}[X]} = \varepsilon$ yields the desired $\delta_0$. Now, consider a more complex scenario where $\mathbb{E}[X]$ is unknown, and only a sample result $\phi$ drawn from the random variable $X$ is available. In this case, we substitute $\mathbb{E}[X]$ with $\phi/(1 + \delta)$, transforming the equation into:
\begin{equation}\label{Eq:chernoffluc}
	e^{-\frac{\delta^2}{2 + \delta} \frac{\phi}{1 + \delta}} = \varepsilon.
\end{equation}
When $\phi = (1 + \delta_0)\mathbb{E}[X]$, it is clear that the solution is exactly $\delta_0$. Since the function $e^{-\frac{\delta^2}{2 + \delta} \frac{1}{1 + \delta}}$ is monotonically decreasing for any $\delta > 0$, when $\phi < (1 + \delta_0)\mathbb{E}[X]$, which occurs with a probability greater than $1 - \varepsilon$, the solution will be larger than $\delta_0$, making it a valid bound for failure probability $\varepsilon$. Therefore, we have the following lemma.

\begin{lemma}[Symmetric Analytical Bound]\label{lemma:analytical}
	Suppose $X = \sum_{i=1}^{n} X_i$ is a random variable with each $X_i$ taking values in $[0, 1]$. Given a security parameter $0 < \varepsilon < 1$ and a sample result $\phi$ drawn from $X$, suppose $\phi > -6\ln\varepsilon$. Then, with a probability greater than $1 - \varepsilon$, the solution to Eq.~\eqref{Eq:chernoffluc}, which has the following form,
	\begin{equation}\label{eq:delta}
		\delta = \frac{-3\ln\varepsilon + \sqrt{\ln^2\varepsilon - 8\phi\ln\varepsilon}}{2(\phi + \ln\varepsilon)},
	\end{equation}
	satisfies
	\begin{equation}
		\Pr\left(\frac{X}{1 + \delta} \geq \mathbb{E}[X]\right) \leq \varepsilon,
	\end{equation}
	and
	\begin{equation}
		\Pr\left(\frac{X}{1 - \delta} \leq \mathbb{E}[X]\right) \leq \varepsilon.
	\end{equation}
\end{lemma}

The condition $\phi > -6\ln\varepsilon$ is needed to ensure that $\delta$ remains within the range $[0, 1]$. When this condition is not met, this analytical bound is no longer applicable. In such a case, we must revert to the original version in Lemma~\ref{lemma:chernoff}. In this case, $\delta$ becomes asymmetric for the lower and upper bounds, and no analytical solution is available. However, due to the monotonicity of the expressions, we can use binary search to solve the following equations:
\begin{equation}
    \begin{split}
        &\left(\frac{e^{\delta_l}}{(1+\delta_l)^{1+\delta_l}}\right)^{\frac{\phi}{1+\delta_l}}=\varepsilon\\
        &\left(\frac{e^{-\delta_u}}{(1-\delta_u)^{1-\delta_u}}\right)^{\frac{\phi}{1+\delta_u}}=\varepsilon,
    \end{split}
\end{equation}
to obtain $\delta_l$ and $\delta_u$. 

Furthermore, we observe that as $\phi$ approaches $-6\ln\varepsilon$, the result for $\delta$ in Eq.~\eqref{eq:delta} becomes less tight and can be up to twice as large as the asymmetric values $\delta_l,\delta_u$ obtained through binary search. Empirical observations indicate that when $\phi > -100\ln\varepsilon$, Eq.~\eqref{eq:delta} provides a sufficiently tight bound. 

\section{Bounds from Linear Equations}\label{Sec:Linear}
In QKD, the information reconciliation term can be directly obtained after Alice and Bob complete the reconciliation process. The core challenge of security analysis lies in determining the privacy amplification term from Eq.~\eqref{Eq:KeyRate}. For simplicity, we focus on the one-decoy state method~\cite{Ma2005practical}, which utilizes two intensities: $\mu$ and $\nu$, where $\mu > \nu$. However, similar derivations can be extended to more general scenarios involving multiple decoy states.

With phase-randomized coherent state sources, as described in Eq.~\eqref{eq:photonnumber}, the photon number channel model can be assumed in the security analysis~\cite{Ma2008PhD}. For each intensity, Alice and Bob can derive a set of linear equations as in Eq.~\eqref{eq:QEeqs}. The coefficients of these linear equations are governed by the Poisson distribution ${\mu^i e^{-\mu}}/{i!}$ for each photon number $i$. This method can be adapted for other photon sources~\cite{Ma2008PhD}, such as thermal states, to accommodate different experimental setups.

\subsection{General Framework}\label{Subsec:framework}
In the development of the decoy-state method, our goal is to solve the following minimization problem:
\begin{equation}\label{Eq:minPAterm}
	\begin{aligned}
		\min_{Y_1, e_1} \quad Y_1\left[1 - h\left(e_1\right)\right]
	\end{aligned}
\end{equation}
subject to the constraints $Y_i, e_i \in [0,1]$ for all $i$, and:
\begin{equation}\label{Eq:linearConstr}
	\begin{aligned}
		Q_{\mu} &= e^{-\mu} Y_0 + \mu e^{-\mu} Y_1 + \frac{\mu^2}{2} e^{-\mu} Y_2 + \frac{\mu^3}{3!} e^{-\mu} Y_3 + \cdots,\\		
		Q_{\nu} &= e^{-\nu} Y_0 + \nu e^{-\nu} Y_1 + \frac{\nu^2}{2} e^{-\nu} Y_2 + \frac{\nu^3}{3!} e^{-\nu} Y_3 + \cdots,\\
		E_{\mu} Q_{\mu} &= e^{-\mu} Y_0 e_0 + \mu e^{-\mu} Y_1 e_1 + \frac{\mu^2}{2} e^{-\mu} Y_2 e_2 + \frac{\mu^3}{3!} e^{-\mu} Y_3 e_3 + \cdots,\\
		E_{\nu} Q_{\nu} &= e^{-\nu} Y_0 e_0 + \nu e^{-\nu} Y_1 e_1 + \frac{\nu^2}{2} e^{-\nu} Y_2 e_2 + \frac{\nu^3}{3!} e^{-\nu} Y_3 e_3 + \cdots.\\
	\end{aligned}
\end{equation}
For simplicity, we use the one-decoy method as an example, but in practice, more decoy states could be used to add additional linear constraints, such as incorporating a vacuum decoy state.

Our goal is to bound $Y_1$ and $e_1$ to determine the privacy amplification term and thus the key rate. Two essential equations help in bounding the single-photon components:
\begin{equation}
	\begin{aligned}
		q_{\mu} &= y_0 + \mu y_1 + \frac{\mu^2}{2} y_2 + \frac{\mu^3}{3!} y_3 + \frac{\mu^4}{4!} y_4 + \cdots,\\
		q_{\nu} &= y_0 + \nu y_1 + \frac{\nu^2}{2} y_2 + \frac{\nu^3}{3!} y_3 + \frac{\nu^4}{4!} y_4 + \cdots,
	\end{aligned}
\end{equation}
where $\mu$ and $\nu$ represent the intensities of the signal and decoy states, respectively, with $\mu > \nu$. Here, $q$ represents measurable quantities like gains or error rates, while $y_i$ are the unknown variables potentially influenced by Eve.

Since the two-photon term is the next largest component after the single-photon term, we can eliminate it, leading to the following equation:
\begin{equation}
	\begin{aligned}
		\mu^2 q_{\nu} - \nu^2 q_{\mu} &= (\mu^2 - \nu^2) y_0 + \mu\nu(\mu - \nu) y_1 - \frac{\mu^2 \nu^2 (\mu - \nu)}{3!} y_3 - \frac{\mu^2 \nu^2 (\mu^2 - \nu^2)}{4!} y_4 - \cdots.
	\end{aligned}
\end{equation}
Our primary interest is in the unknown term $y_1$, which is expressed as:
\begin{equation} \label{eq:y1est}
	y_1 = \frac{1}{\mu\nu(\mu - \nu)} \left[ \mu^2 q_{\nu} - \nu^2 q_{\mu} - (\mu^2 - \nu^2) y_0 + \frac{\mu^2 \nu^2 (\mu - \nu)}{3!} y_3 + \frac{\mu^2 \nu^2 (\mu^2 - \nu^2)}{4!} y_4 + \cdots \right].
\end{equation}
If three different intensities are used, the three-photon term can also be eliminated. More generally, using $k$ different intensities, we can eliminate up to $k$-photon terms. As the number of decoy states increases, the remaining terms become exponentially smaller, allowing for tighter bounds on $y_1$ and better security of the final key rate~\cite{Yuan2016Simulating}.

\subsection{Previous Results}
In the earlier work on the vacuum+weak decoy-state method~\cite{Ma2005practical}, the gains for different intensities can be expressed as $q_\mu = Q_{\mu} e^{\mu}$ and $q_\nu = Q_{\nu} e^{\nu}$, which can be substituted into Eq.~\eqref{eq:y1est}. Additionally, there is a constraint from the vacuum decoy state that ensures $y_i = Y_i \ge 0$. Using these, we can derive a lower bound for $Y_1$ as follows:
\begin{equation} \label{eq:oldY1Low}
	Y_1 \ge \frac{1}{\mu \nu (\mu - \nu)} \left[ \mu^2 Q_{\nu} e^{\nu} - \nu^2 Q_{\mu} e^{\mu} - (\mu^2 - \nu^2) Y_0 \right],
\end{equation}
where $q_\mu = Q_{\mu} e^{\mu}$ and $q_\nu = Q_{\nu} e^{\nu}$ have been substituted. This bound can be reached if Eve sets $Y_i = 0$ for all $i \geq 3$. Note that the constraint on $Y_2$, specifically $Y_2 \in [0,1]$, generally holds true in this context.

To bound $e_1$, we use the error rate constraints from Eq.~\eqref{Eq:linearConstr}:
\begin{equation}\label{eq:oldE1}
	e_1 \le \min\left[ \frac{E_{\mu} Q_{\mu} e^{\mu} - Y_0 e_0}{\mu Y_1}, \frac{E_{\nu} Q_{\nu} e^{\nu} - Y_0 e_0}{\nu Y_1} \right].
\end{equation}
Typically, the second term in the minimization is smaller if we do not consider statistical fluctuations. To obtain an upper bound on $e_1$, we substitute the lower bound for $Y_1$ derived above. Note that this bound on $e_1$ tends to be a bit loose. For equality to hold, Eve would need to set $e_i = 0$ for all $i \geq 2$. However, with this setting, Alice and Bob can solve the linear equations for error rates to derive a tighter bound.

\subsection{Improved Estimation}\label{subsec:eve}
In this subsection, we provide an analytical parameter configuration for the minimization problem in Eq.~\eqref{Eq:minPAterm}. While the solution is exact, applying it directly to fluctuation analysis is challenging due to its piecewise nature and the requirement of vacuum decoy states. Nevertheless, this exact form offers valuable insights, motivating subsequent derivations, as demonstrated in the next subsection.

\begin{theorem}[Eve's Optimal Strategy]\label{thm:best}
	Assume $Y_0$ is a fixed constant. To minimize $Y_1[1-h(e_1)]$, Eve should set $e_i=1$ for any $i\geq 3$. For the yield $Y_i$, let $k_0$ be the smallest integer such that
	\begin{equation}
		\sum_{i=k_0}^{\infty}\frac{\mu^{i-1}-\nu^{i-1}}{i!}\leq\frac{1}{\mu}e^{\mu}E_{\mu}Q_{\mu}-\frac{1}{\nu}e^{\nu}E_{\nu}Q_{\nu}+\frac{\mu-\nu}{2\mu\nu}Y_0,
	\end{equation}
	then the configuration follows:
	\begin{equation}
		Y_i=\left\{\begin{aligned}
			&0, & 3\leq &k\leq k_0-2\\
			&\frac{k!}{\mu^{k-1}-\nu^{k-1}}\left(\frac{1}{\mu}e^{\mu}E_{\mu}Q_{\mu}-\frac{1}{\nu}e^{\nu}E_{\nu}Q_{\nu}+\frac{\mu-\nu}{2\mu\nu}Y_0-\sum_{i=k_0}^{\infty}\frac{\mu^{i-1}-\nu^{i-1}}{i!}\right),& &k=k_0-1\\
			&1,& &k\geq k_0
		\end{aligned}\right..
	\end{equation}
\end{theorem}

\begin{proof}
	As is often done in the decoy-state method, we first eliminate $e_2, Y_2$ and express $e_1, Y_1$ as functions of other variables:
	\begin{equation}
		\begin{split}
			Y_1&=\frac{1}{\mu\nu(\mu-\nu)}\left[\Bigl(\mu^2e^{\nu}Q_{\nu}-\nu^2e^{\mu}Q_{\mu}\Bigr)-\Bigl((\mu^2-\nu^2)Y_0\Bigr)+\Bigl(\mu^2\nu^2\sum_{i=3}^{\infty}\frac{\mu^{i-2}-\nu^{i-2}}{i!}Y_i\Bigr)\right],\\
			e_1&=\frac{\Bigl(\mu^2e^{\nu}E_{\nu}Q_{\nu}-\nu^2e^{\mu}E_{\mu}Q_{\mu}\Bigr)-\Bigl((\mu^2-\nu^2)e_0Y_0\Bigr)+\Bigl(\mu^2\nu^2\sum_{i=3}^{\infty}(\mu^{i-2}-\nu^{i-2})e_iY_i/i!\Bigr)}{\Bigl(\mu^2e^{\nu}Q_{\nu}-\nu^2e^{\mu}Q_{\mu}\Bigr)-\Bigl((\mu^2-\nu^2)Y_0\Bigr)+\Bigl(\mu^2\nu^2\sum_{i=3}^{\infty}(\mu^{i-2}-\nu^{i-2})Y_i/i!\Bigr)}.
		\end{split}
	\end{equation}
	
Define several intermediate variables:
\begin{equation}\label{Eq:definitions}
	\begin{split}
		C_1&=\mu^2e^{\nu}Q_{\nu}-\nu^2e^{\mu}Q_{\mu},\\
		C_2&=\mu^2e^{\nu}E_{\nu}Q_{\nu}-\nu^2e^{\mu}E_{\mu}Q_{\mu},\\
		z&=(\mu^2-\nu^2)e_0Y_0,\\
		x&=\mu^2\nu^2\sum_{i=3}^{\infty}(\mu^{i-2}-\nu^{i-2})Y_i/i!,\\
		y&=\mu^2\nu^2\sum_{i=3}^{\infty}(\mu^{i-2}-\nu^{i-2})e_iY_i/i!.
	\end{split}
\end{equation}
Here, $e_0=1/2$, and $C_1, C_2$ are constants. Since $Y_0$ is fixed, $z$ is also a constant and the objective to minimize becomes a function of $x, y$:
\begin{equation}
	Y_1[1-h(e_1)]=\frac{F(x,y)}{\mu\nu(\mu-\nu)}=\frac{C_1-2z+x}{\mu\nu(\mu-\nu)}\left[1-h\left(\frac{C_2-z+y}{C_1-2z+x}\right)\right].
\end{equation}

To determine the monotonicity of this function, we find the partial derivatives:
\begin{equation}
	\begin{split}
		\frac{\partial{F}}{\partial{x}}&=1+\log(1-e_1),\\
		\frac{\partial{F}}{\partial{y}}&=\log(e_1)-\log(1-e_1).
	\end{split}
\end{equation}
When $e_1>1/2$, there will be no key generated by QKD. Hence, we are only interested in the domain of $e_1<1/2$. Then, $\frac{\partial{F}}{\partial{x}} > 0$, $\frac{\partial{F}}{\partial{y}} < 0$, and $\frac{\partial{F}}{\partial{z}} > 0$ always hold. Therefore, the best strategy for Eve is to maximize $y$ and minimize $x$. Since $0\leq y\leq x$, the optimal choice is to set $y=x$, i.e., $e_i=1$ for all $i\geq 3$. 

Under this condition, the function $F$ only depends on $x$:
\begin{equation}\label{Eq:withouty}
	F=(C_1-2z+x)\left[1-h\left(\frac{C_2-z+x}{C_1-2z+x}\right)\right].
\end{equation}
We compute the derivative of $F$ with respect to $x$:
\begin{equation}
	\frac{\mathrm{d}F}{\mathrm{d}x}=1+\log(e_1).
\end{equation}
Since this derivative is always negative for $e_1 < 1/2$, Eve should maximize $x$. When $x$ grows, $e_1,Y_1$ will increase simultaneously. To keep the fourth line in Eq.~\eqref{Eq:linearConstr} holds, $e_2Y_2$ will inevitably decrease. As a result, $x$ reaches its maximum when $e_2Y_2=0$. Combining the third and fourth lines of Eq.~\eqref{Eq:linearConstr} to eliminate $e_1Y_1$, we obtain:
\begin{equation}\label{Eq:x_constr}
	\sum_{i=3}^{\infty}\frac{\mu^{i-1}-\nu^{i-1}}{i!}Y_i=\frac{1}{\mu}e^{\mu}E_{\mu}Q_{\mu}-\frac{1}{\nu}e^{\nu}E_{\nu}Q_{\nu}+\frac{\mu-\nu}{2\mu\nu}Y_0.
\end{equation}
Since $\frac{\mu^{i-2}-\nu^{i-2}}{\mu^{i-1}-\nu^{i-1}}$ is a monotonically increasing series, the strategy to maximize $x$ is to set all $Y_i=1$ for sufficiently large $i$ until a point $k$ is reached such that for some $Y_k \leq 1$, Eq.~\eqref{Eq:x_constr} is satisfied.
\end{proof}

At first glance, this result may seem counter-intuitive. Why would Eve willingly increase errors for $e_i$ when $i \geq 3$, considering they contribute nothing to hacking the final key rate? The rationale lies in the decoy-state method, where we eliminate variables $e_2$ and $Y_2$, providing no advantage to Eve in manipulating $e_2$. Attempting to minimize $e_2$ indirectly increases $e_1$, which decreases the final key rate. Consequently, Eve is compelled to increase $e_i$ and $Y_i$ for $i \geq 3$ to satisfy the linear constraints in Eq.~\eqref{Eq:linearConstr}.

The assumption that $Y_0$ is fixed holds when Alice and Bob adopt a vacuum decoy state to directly measure $Y_0$. However, in the one-decoy-state method, $Y_0$ is uncertain. Interestingly, the minimal value of $Y_1[1-h(e_1)]$ has a complex relationship with $Y_0$, depending on the values of $\mu$ and $\nu$. In the next subsection, we explore how linear expansion can simplify this problem.

\subsection{Linear Expansion}\label{subsec:improved}
To solve the minimization problem of Eq.~\eqref{Eq:minPAterm}, it is inefficient to bound $Y_1$ and $e_1$ separately. Instead, we can directly bound Eq.~\eqref{Eq:minPAterm} by transforming it into a linear function of $e_1$. Since $1-h(e)$ is a convex function in $e$, it can be lower bounded by any tangent line to the curve. For all $e_t\in(0,1)$, we define the parameters as:
\begin{equation} 
	\begin{split}
		a &= \log (1-e_t) + 1, \\
		b &= \log(1-e_t) - \log{e_t},
	\end{split}
\end{equation} 
such that the following inequality holds:
\begin{equation}\label{Eq:expansion}
	\begin{split}
		1-h(e) &\ge a-be.
	\end{split}
\end{equation}
As shown in Figure \ref{fig:he}, the tangent line at any point $e_t \in (0,1)$ provides a lower bound for $1-h(e)$. Here, we choose $a$ and $b$ such that the line $f(e) = a - b e$ is the tangent line at the point $e = e_t$.

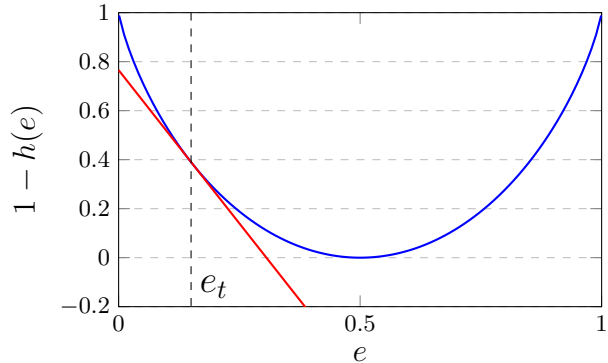
\begin{figure}[htbp!]
	\begin{tikzpicture}
		\begin{axis}[
			scale=.7,
			width=0.6\linewidth,
			height=0.4\linewidth,
			xlabel={\large $e$},
			ylabel={\large $1-h(e)$},
			xmin=0, xmax=1,
			ymin=-.2, ymax=1,
			xtick={0,0.5,1},
			ytick={-0.2,0,0.2,0.4,0.6,0.8,1},
			legend pos=north west,
			ymajorgrids=true,
			grid style=dashed,
			]
			\addplot[domain=0.001:0.999, samples=100, color=blue, thick] {1 + x*ln(x)/ln(2) + (1-x)*ln(1-x)/ln(2)};
			\addplot[domain=0:1, samples=2, color=red, thick] {0.76553 - 2.5025*x};
			\draw[dashed] (axis cs:0.15,1)--(axis cs:0.15,-0.2) node[at end, above right] {\Large $e_t$};
		\end{axis}
	\end{tikzpicture}
	\caption{Plot of $1-h(e)$ and its tangent line at $e=e_t$.} \label{fig:he}
\end{figure}

Treating Eq.~\eqref{Eq:withouty} as a function of $x$ and $z$ and employing Eq.~\eqref{Eq:expansion}, we have:
\begin{equation}\label{Eq:linearelax}
	\begin{split}
		Y_1[1-h(e_1)]&\geq\frac{C_1-2z+x}{\mu\nu(\mu-\nu)}\left[1-h\left(\frac{C_2-z+x}{C_1-2z+x}\right)\right],\\
		&\geq\frac{C_1-2z+x}{\mu\nu(\mu-\nu)}\left[a-b\cdot\frac{C_2-z+x}{C_1-2z+x}\right],\\
		&\geq\frac{1}{\mu\nu(\mu-\nu)}\Bigl[aC_1-bC_2+(a-b)\nu(\nu e^{\mu}E_{\mu}Q_{\mu}-\mu e^{\nu}E_{\nu}Q_{\nu})\Bigr].
	\end{split}
\end{equation}
The last line follows because:
\begin{equation}
	\begin{split}
		x&<\mu\nu^2\sum_{i=3}^{\infty}\frac{\mu^{i-1}-\nu^{i-1}}{i!}Y_i,\\
		&=\nu\Bigl(\nu e^{\mu}E_{\mu}Q_{\mu}-\mu e^{\nu}E_{\nu}Q_{\nu}\Bigr)+\frac{\nu}{\mu+\nu}z,
	\end{split}
\end{equation}
from Eq.~\eqref{Eq:x_constr} and, by appropriately selecting $a$ and $b$, we can ensure that $(a - b)\frac{\nu}{\mu + \nu} + (b - 2a)$ is non-negative.

We can interpret this lengthy expression in a more explicit way. Tracing back to the third line in Eq.~\eqref{Eq:linearelax}, we neglect the term $(a-b)\nu(\nu e^{\mu}E_{\mu}Q_{\mu}-\mu e^{\nu}E_{\nu}Q_{\nu})$ derived from $x$ and $z$ and focus solely on $aC_1 - bC_2$. This result can be derived using the general framework introduced in Section \ref{Subsec:framework} by setting:
\begin{equation} 
	\begin{split}
		q_{\mu} &= aQ_{\mu}e^{\mu} -bE_{\mu}Q_{\mu} e^{\mu}, \\
		q_{\nu} &= aQ_{\nu}e^{\nu} -bE_{\nu}Q_{\nu} e^{\nu}, \\
		y_i &= aY_i-bY_ie_i.
	\end{split}
\end{equation}
Since for the vacuum state $e_0 = 1/2$, it is easy to verify that $y_0 \geq 0$. If $y_i \geq 0$ holds for all $i \geq 3$, then using Eq.~\eqref{eq:y1est}, we obtain:
\begin{equation}\label{eq:ImprovedEstimation}
	\begin{split}
		Y_1[1-h(e_1)] &\ge y_1 \\
		&\stackrel{?}{\geq} \frac{1}{\mu\nu(\mu-\nu)} \left[\mu^2q_{\nu}-\nu^2q_{\mu} \right] \\
		&= \frac{1}{\mu\nu(\mu-\nu)} \left[\mu^2 (aQ_{\nu}e^{\nu} -bE_{\nu}Q_{\nu} e^{\nu})-\nu^2 (aQ_{\mu}e^{\mu} -bE_{\mu}Q_{\mu} e^{\mu}) \right].
	\end{split}
\end{equation}
This is the term $\frac{1}{\mu\nu(\mu - \nu)}(aC_1 - bC_2)$. However, the greater-than sign in the second line is questionable since the assumption $y_i \geq 0$ does not generally hold. In fact, as shown in Figure \ref{fig:he}, this fails for all $i \geq 3$ under Eve's optimal strategy, where Eve sets $e_i = 1$ for all $i \geq 3$. This failure is why the term $(a - b)\nu(\nu e^{\mu}E_{\mu}Q_{\mu} - \mu e^{\nu}E_{\nu}Q_{\nu})$ appears in the final linear expression in Eq.~\eqref{Eq:linearelax}. With this additional term, the bound becomes exact.

The reason for dividing Eq.~\eqref{Eq:linearelax} into two terms is to provide an intuitive understanding: compared to $aC_1 - bC_2$, the additional term is typically small and can be regarded as a correction for a simpler setting where all $e_i = 0$ and $Y_i = 0$ for $i \geq 3$. This simpler setting is easier to handle, involving only a finite number of non-zero variables, and the result derived from it, after applying corrections, is quite satisfactory. We formalize this interpretation in the following theorem.

\begin{theorem}[Asymptotic Bound]\label{thm:simple}
	Suppose the intensities of the signal and decoy states are $\mu$ and $\nu$, respectively. In the asymptotic limit, the gain and bit error rate corresponding to the signal and decoy states are $Q_{\mu}, E_{\mu}$ and $Q_{\nu}, E_{\nu}$. Then the privacy amplification term can be lower bounded by:
	\begin{equation}\label{eq:keyrateboundasym}
		Y_1[1-h(e_1)] \ge Y_1^*[1-h(e_1^*)]+\frac{\nu}{2}(1+\log e_1^*)e_2^*Y_2^*,
	\end{equation}
	where $e_1^*, Y_1^*, e_2^*, Y_2^*$ are the solution to equations:
    \begin{equation}\label{eq:simple-setting}
        \begin{split}
            Q_{\mu}&=\mu e^{-\mu}Y_1^*+\frac{\mu^2}{2} e^{-\mu}Y_2^*,\\
            Q_{\nu}&=\nu e^{-\nu}Y_1^*+\frac{\nu^2}{2} e^{-\nu}Y_2^*,\\
            E_{\mu}Q_{\mu}&=\mu e^{-\mu}Y_1^*e_1^*+\frac{\mu^2}{2} e^{-\mu}Y_2^*e_2^*,\\
            E_{\nu}Q_{\nu}&=\nu e^{-\nu}Y_1^*e_1^*+\frac{\nu^2}{2} e^{-\nu}Y_2^*e_2^*.
        \end{split}
    \end{equation}             
\end{theorem}

\begin{proof}
	We determine the value of the tangent point $e_t$ and apply Eq.~\eqref{Eq:linearelax} to obtain this result. First, consider the scenario where $Y_0=Y_i=0$ for all $i \geq 3$. In this case, there are only four unknown variables: $Y_1, Y_2, e_1, e_2$. The constraint from Eq.~\eqref{Eq:linearConstr} can be reduced to Eq.~\eqref{eq:simple-setting} and therefore becomes solvable. Its solution is given by Eq.~\eqref{eq:starcase},
 	\begin{equation}\label{eq:starcase}
		\begin{split}
			Y_1^*&=\frac{1}{\mu\nu(\mu-\nu)}(\mu^2Q_{\nu}e^{\nu}-\nu^2Q_{\mu}e^{\mu}),\\
			e_1^*&=\frac{\mu^2E_{\nu}Q_{\nu}e^{\nu}-\nu^2E_{\mu}Q_{\mu}e^{\mu}}{\mu^2Q_{\nu}e^{\nu}-\nu^2Q_{\mu}e^{\mu}},\\
			Y_2^*&=\frac{2}{\mu\nu(\mu-\nu)}(\nu Q_{\mu}e^{\mu}-\mu Q_{\nu}e^{\nu}),\\
			e_2^*&=\frac{\nu E_{\mu}Q_{\mu}e^{\mu}-\mu E_{\nu}Q_{\nu}e^{\nu}}{\nu Q_{\mu}e^{\mu}-\mu Q_{\nu}e^{\nu}}.
		\end{split}
	\end{equation}   
     Then, we set $a=1+\log(1-e_1^*)$ and $b=\log(1-e_1^*)-\log e_1^*$, which yields $1-h(e_1^*)=a-be_1^*$. Substituting these relations into Eq.~\eqref{Eq:linearelax} gives Eq.~\eqref{eq:keyrateboundasym}.
\end{proof}

One important point to note is ensuring that the chosen values of $a$ and $b$ do not violate the condition $(a-b)\frac{\nu}{\mu+\nu}+(b-2a)>0$. As we will discuss in Sec.~\ref{subsec:justification}, this condition holds in most cases. In extreme cases where this condition is violated, we must alter our choice of $a$ and $b$. Typically, this can be achieved by selecting an $e_t$ smaller than $e_1^*$. Later, in Section \ref{Sec:Simulation}, numerical results will support our claim that the term $\frac{\nu}{2}(1+\log e_1^*)e_2^*Y_2^*$ is generally smaller in magnitude than $Y_1^*[1-h(e_1^*)]$ in most practical scenarios and can be considered a minor correction term.

\section{Statistical Fluctuation Analysis}\label{Sec:Joint}
Consider the definition of decoy-state statistics from Box \ref{box:decoy}. The number of states prepared are denoted by $N_{\mu}$ and $N_{\nu}$, which may deviate from their expected values $Np_{\mu}$ and $Np_{\nu}$, respectively. After state preparation, Alice sends these states through the channel. In the channel, these states are categorized according to their photon numbers, resulting in $N_0, N_1, \cdots$, whose actual values can differ from the expected numbers, given by $N_{\mu}\cdot e^{-\mu}\mu^i/i! + N_{\nu}\cdot e^{-\nu}\nu^i/i!$. 

The malicious Eve can observe these categories and manipulate the states accordingly. After her intervention, the number of valid detection events in each category can be predicted, denoted by $n_0, n_1, \cdots$. Bob, on receiving these states and performing measurements, records the clicks as $n_{\mu}, n_{\nu}$, which he will use to derive the keys. 

The transition from $n_0, n_1, n_2, \cdots$ to $n_{\mu}, n_{\nu}$ is also subject to fluctuations. This is because Eve cannot distinguish between states originating from different intensities, leading to the decomposition $n_i = n_{i, \mu} + n_{i, \nu}$ having numerous possible configurations. As a result, the total clicks seen by Bob, $n_{\mu} = n_{0, \mu} + n_{1, \mu} + \cdots$, can also vary across different possibilities. These intricacies significantly complicate the fluctuation analysis.

A previous study~\cite{Zhang2017improved} tackled this issue by assuming that $n_i$ is fixed but unknown, which greatly simplified the analysis. Building on this idea, we extend the fluctuation analysis by accounting for the correlations between clicks from different intensities, thereby establishing a more comprehensive framework. We apply this enhanced analysis to the results obtained in the previous section. Notations that will be widely used in this section are listed in Table~\ref{table:Notation}.

\begin{center}
	\begin{table}[htbp]
		\caption{Notations}\label{table:Notation}
		\centering
		\begin{tabular}{cc}\hline
			Quantity &  Notation\\\hline
			$\mu$ & Intensity for signal states\\\hline
			$\nu$ & Intensity for decoy states\\\hline
			$p_{\mu}$ & Probability of sending signal states\\\hline
			$p_{\nu}$ & Probability of sending decoy states\\\hline
			$N$ &  Total number of laser pulses emitted\\\hline
			$N_{\mu}$ & Total number of laser pulses emitted in signal states\\\hline
			$N_{\nu}$ & Total number of laser pulses emitted in decoy states\\\hline
            $N_{\mu,1}$ & Total single photon pulses emitted in decoy states\\\hline
			$n_{\mu}$ & Valid clicks in signal states\\\hline
            $n_{\mu,1}$ & Valid single photon clicks in signal states\\\hline
			$n_{\nu}$ & Valid clicks in decoy states\\\hline
			$m_{\mu}$ & Bit errors in signal states\\\hline
			$m_{\nu}$ & Bit errors in decoy states\\\hline
			$c_{\mu}$ & Valid clicks without bit errors in signal states\\\hline
			$c_{\nu}$ & Valid clicks without bit errors in decoy states\\\hline 
            $Y$ & Privacy amplification term $Y_1[1-h(e_1)]$\\\hline
            $R$ & Key rate per emitted pulse\\\hline
		\end{tabular}
	\end{table}
\end{center}

\subsection{Random variable formulation}
To understand this formulation, we first examine a scenario without fluctuations, where the relation between valid clicks sent from different sources, denoted by $n_{\mu}, n_{\nu}$, and valid clicks obtained from different photon-number channels, explicitly represented as $n_0, n_1, \cdots$, is as follows:
\begin{equation}\label{eq:noFluc}
	\begin{split}
		n_{\mu} &= p_{\mu|0}n_0 + p_{\mu|1}n_1 + p_{\mu|2}n_2 + \cdots,\\
		n_{\nu} &= p_{\nu|0}n_0 + p_{\nu|1}n_1 + p_{\nu|2}n_2 + \cdots,
	\end{split}
\end{equation}
where the conditional probability is given by
\begin{equation}
	p_{\mu|i} = \frac{p_{\mu} p_{\mu,i}}{p_{\mu} p_{\mu,i} + p_{\nu} p_{\nu,i}},
\end{equation}
with $p_{\mu,i}$ following the Poisson distribution:
\begin{equation}
	p_{\mu,i} = e^{-\mu} \frac{\mu^i}{i!}.
\end{equation}

In addition to $n_{\mu}, n_{\nu}$, the count of valid clicks with bit errors, $m_{\mu}, m_{\nu}$, can also be obtained from the experiment. Likewise, their connections with $m_i = n_i e_i$ are as follows:
\begin{equation}\label{eq:RnoFluc}
	\begin{split}
		m_{\mu} &= p_{\mu|0} m_0 + p_{\mu|1} m_1 + p_{\mu|2} m_2 + \cdots,\\
		m_{\nu} &= p_{\nu|0} m_0 + p_{\nu|1} m_1 + p_{\nu|2} m_2 + \cdots.
	\end{split}
\end{equation}

Note that all $n_0, m_0, n_1, m_1, \cdots$ are fixed but unknown. The fluctuation arises solely from the procedure of assigning $n_i, m_i$ to different intensities. We assume that Alice chooses the coherent states independently in each round. The core of the decoy-state method is that Eve has no information about the intensities chosen by Alice when controlling the values of $n_i, m_i$. Therefore, the assignment procedure, taking $n_i$ as an example, can be regarded as $n_i$ i.i.d.~random variables $\chi_{i,1}, \cdots, \chi_{i, n_i}$ with a probability of $p_{\mu|i}$ to contribute one click to $n_{\mu}$, and a probability of $p_{\nu|i}$ to contribute one click to $n_{\nu}$.

Mathematically, if we define $\chi_{i,j}$ in the following way:
\begin{equation}
	\chi_{i,j} = \begin{cases}
		\mu, & \text{with probability } p_{\mu|i},\\
		\nu, & \text{with probability } p_{\nu|i},
	\end{cases}
\end{equation}
the number of bit errors $m_{\mu}, m_{\nu}$ and clicks $n_{\mu}, n_{\nu}$ can be written as:
\begin{equation}\label{eq:rvVersion}
	\begin{split}
		m_{\mu} &= \sum_{j=1}^{m_0} 1_{\chi_{0,j} = \mu} + \sum_{j=1}^{m_1} 1_{\chi_{1,j} = \mu} + \sum_{j=1}^{m_2} 1_{\chi_{2,j} = \mu} + \cdots,\\
		m_{\nu} &= \sum_{j=1}^{m_0} 1_{\chi_{0,j} = \nu} + \sum_{j=1}^{m_1} 1_{\chi_{1,j} = \nu} + \sum_{j=1}^{m_2} 1_{\chi_{2,j} = \nu} + \cdots,\\
		n_{\mu} &= m_{\mu} + \sum_{j = m_0 + 1}^{n_0} 1_{\chi_{0,j} = \mu} + \sum_{j = m_1 + 1}^{n_1} 1_{\chi_{1,j} = \mu} + \sum_{j = m_2 + 1}^{n_2} 1_{\chi_{2,j} = \mu} + \cdots,\\
		n_{\nu} &= m_{\nu} + \sum_{j = m_0 + 1}^{n_0} 1_{\chi_{0,j} = \nu} + \sum_{j = m_1 + 1}^{n_1} 1_{\chi_{1,j} = \nu} + \sum_{j = m_2 + 1}^{n_2} 1_{\chi_{2,j} = \nu} + \cdots,
	\end{split}
\end{equation}
where $1_{(\cdot)}$ is the indicator function. In the absence of fluctuations, these equations revert to Eq.~\eqref{eq:noFluc}. The source of fluctuation lies in the Bernoulli-like random variables $\chi_{i,j}$, which can be tightly bounded using the Chernoff bounds. For later convenience, we define $c_{\mu} = n_{\mu} - m_{\mu}, c_{\nu} = n_{\nu} - m_{\nu}$, and $c_{i} = n_{i} - m_{i}$ to represent the number of clicks without errors.

\subsection{Constructing lower bound estimator}
In the last section, we derived a lower bound for the key rate in Eq.~\eqref{Eq:linearelax}. In this subsection, we employ the random variable formulation above to construct a lower bound for $Y_1[1-h(e_1)]$. Recall the definition of gain $Q$, the right-hand side of Eq.~\eqref{Eq:linearelax} can be rewritten as:
\begin{equation}\label{Eq:prac1}
	\begin{split}
		\hat{Y}=\frac{1}{\mu\nu(\mu-\nu)} \left[\mu^2 \left( a\frac{n_{\nu}}{N_{\nu}}e^{\nu} -b\frac{m_{\nu}}{N_{\nu}} e^{\nu}\right)-\nu^2 \left(a\frac{n_{\mu}}{N_{\mu}}e^{\mu} -b\frac{m_{\mu}}{N_{\mu}} e^{\mu}\right) + (a-b)\nu\left( \nu e^{\mu}\frac{m_{\mu}}{N_{\mu}}-\mu e^{\nu}\frac{m_{\nu}}{N_{\nu}}\right) \right].
	\end{split}
\end{equation}

We use Eq.~\eqref{eq:rvVersion} to substitute $n_{\mu}, n_{\nu}, m_{\mu}, m_{\nu}$ with random variables $\chi_{i,j}$ and reformulate this expression:
\begin{equation}
	\begin{split}
		\hat{Y} &= \frac{1}{\mu\nu(\mu-\nu)}\sum_{k=1}^{2}\sum_{i=0}^{\infty}\sum_{j=1}^{l_{k,i}} T_{k,ij}
	\end{split}
\end{equation}
with $l_{1,i}=m_i, l_{2,i}=c_{i}$ and:
\begin{equation} \label{Eq:XandY}
	\begin{split}
		T_{1,ij} &= (a-b)\frac{e^{\nu}}{N_{\nu}}\mu(\mu-\nu)\cdot 1_{\chi_{i,j}=\nu},\\
		T_{2,ij} &= -a\frac{\nu^2e^{\mu}}{N_{\mu}}\cdot 1_{\chi_{i,j+m_i}=\mu} + a\frac{\mu^2e^{\nu}}{N_{\nu}}\cdot 1_{\chi_{i,j+m_i}=\nu}.
	\end{split}
\end{equation}

We know the expectation $\mathbb{E}[\hat{Y}]$ serves as a strict lower bound for $Y_1[1-h(e_1)]$. However, it should be pointed out that the probability of this lower bound estimator $\hat{Y}$ being violated may be significant due to fluctuations. To derive a practically useful lower bound $Y^L$ from $\hat{Y}$, we first decompose each $T_k$ into $T_k = \Delta_k \cdot (Z_{k} + B_k)$, where $Z_{k} = \sum_{i=0}^{\infty}\sum_{j=1}^{l_{k,i}} Z_{k,ij}$ are sums of random variables taking values in $[0,1]$, and $\Delta_k, B_k$ are constants. At this stage, all fluctuations occur on variable $Z_{k,ij}$. Subsequently, according to whether $\Delta_k$ is positive or not, a factor $\frac{1}{1+\text{sgn}(\Delta_k)\cdot\delta_k}$ is multiplied to all $Z_{k,ij}$. The failure probability can then be bound as follows:
\begin{equation}
	\begin{split}
		\Pr\bigl(Y_1[1-h(e_1)] < Y^{L}\bigr) &\leq \Pr(\mathbb{E}[\hat{Y}] < Y^{L}) \\
		&= \Pr\left(\mathbb{E}\left[\sum_{k,i,j} \Delta_k \cdot Z_{k,ij}\right] \leq \sum_{k,i,j}\frac{\Delta_k \cdot Z_{k,ij}}{1+\text{sgn}(\Delta_k)\cdot \delta_k}\right) \\
		&\leq \sum_k \Pr\left(\text{sgn}(\Delta_k)\cdot\mathbb{E}\left[\sum_{i,j} Z_{k,ij}\right] \leq \frac{\text{sgn}(\Delta_k)}{1+\text{sgn}(\Delta_k)\delta_k} \sum_{i,j} Z_{k,ij}\right),
	\end{split}
\end{equation}
where the first line uses the fact that $\mathbb{E}[\hat{Y}]$ is a strict lower bound for $Y_1[1-h(e_1)]$. Therefore, we can see that the $\delta_k$ should be selected such that the relation 
\begin{equation}
\Pr\left(\text{sgn}(\Delta_k)\cdot\mathbb{E}\left[\sum_{i,j} Z_{k,ij}\right] \leq \frac{\text{sgn}(\Delta_k)}{1+\text{sgn}(\Delta_k)\delta_k} \sum_{i,j} Z_{k,ij}\right) \leq \varepsilon	
\end{equation}
holds. We refer to Sec.~\ref{subsec:Chernoff} for the detailed solution for $\delta_k$. The final reliable lower bound estimator is given by:
\begin{equation}
	Y^L = \frac{1}{\mu\nu(\mu-\nu)}\left(\sum_{k=1}^{2}\sum_{i=0}^{\infty}\sum_{j=1}^{l_{k,i}}\Delta_k \cdot \left(\frac{Z_{k,ij}}{1+\text{sgn}(\Delta_k)\delta_k} + B_k\right)\right).
\end{equation}
Note that the decomposition $T_k=\Delta_k\cdot(Z_k+B_k)$ is not unique. By defining
\begin{equation}
    \begin{split}
        \Delta_k'&=-\Delta_k\\
        B_k'&=-1-B_k\\
        Z_k'&=1-Z_k,
    \end{split}
\end{equation}
it follows that $T_k=\Delta_k'\cdot(Z_k'+B_k')$ is also a valid decomposition. In practical applications, we can choose the decomposition that yields the tightest bounds.

\subsection{Formula summary for experiment}\label{subsec:exp}
In this subsection, we give the detailed formula and derivation of the final lower bound estimator for key generation rate $R$ per emitted pulse. The result is summarized in the following theorem.

\begin{theorem}
	Suppose in a decoy-state QKD system, the experimental settings and results are shown in Table~\ref{table:Notation}. Then $R^L$ in the following formula serves as a lower bound for the key rate with a failure probability less than $3\varepsilon$:	
	\begin{equation}\label{Eq:KeyRatePrac}
		\begin{split}
			R^L &= \frac{\mu e^{-\mu}N_{\mu}}{\mu\nu(\mu-\nu)N(1+\delta_{N})}\Biggl( (a-b)\frac{e^{\nu}}{N_{\nu}}(\mu^2-\nu^2)\cdot\frac{m_{\nu}}{1-\delta_1} \\
			&\quad + a\Bigl(\frac{\mu^2e^{\nu}}{N_{\nu}} + \frac{\nu^2e^{\mu}}{N_{\mu}}\Bigr) \left( \frac{c_{\nu}}{1+\delta_2} - \frac{\nu^2e^{\mu}/N_{\mu}}{\nu^2e^{\mu}/N_{\mu} + \mu^2e^{\nu}/N_{\nu}}(c_{\mu}+c_{\nu}) \right) \Biggr) - I_{ec}\frac{n_{\mu}}{N},
		\end{split}
	\end{equation}
	where $a, b$ satisfy:	
	\begin{equation}
		\begin{split}
			a &= \log(1-e) + 1, \\
			b &= \log(1-e) - \log{e},
		\end{split}
	\end{equation}
	with:	
	\begin{equation}
		\begin{split}
			e &= \frac{\mu^2\frac{m_{\nu}}{Np_{\nu}}e^{\nu} - \nu^2\frac{m_{\mu}}{Np_{\mu}}e^{\mu}}{\mu^2\frac{n_{\nu}}{Np_{\nu}}e^{\nu} - \nu^2\frac{n_{\mu}}{Np_{\mu}}e^{\mu}},
		\end{split}
	\end{equation}
	and $\delta_N, \delta_1, \delta_2$ subject to:	
	\begin{equation}\label{Eq:delta_in_exp}
		\begin{split}
			\delta_N &= \frac{-\ln\varepsilon + \sqrt{-4N_{\mu}\mu e^{-\mu}\ln\varepsilon + \ln^2\varepsilon}}{2N_{\mu}\mu e^{-\mu}}, \\
			\delta_1 &= \frac{-3\ln{\varepsilon} + \sqrt{-8m_{\nu}\ln{\varepsilon} + \ln^2{\varepsilon}}}{2(m_{\nu} + \ln\varepsilon)}, \\
			\delta_2 &= \frac{-3\ln{\varepsilon} + \sqrt{-8c_{\nu}\ln{\varepsilon} + \ln^2{\varepsilon}}}{2(c_{\nu} + \ln\varepsilon)}.
		\end{split}
	\end{equation}
\end{theorem}

\begin{proof}
    Expressed using variables defined in Table~\ref{table:Notation}, the lower bound of key generation rate per emitted pulse is,
    \begin{equation}
        \begin{split}
        R&=\frac{n_{\mu, 1}}{N}[1-h(e_1)]-fh(E_{\mu})\frac{n_{\mu}}{N},\\
        &\geq\frac{N_{\mu,1}^LY^L}{N}-fh(E_{\mu})\frac{n_{\mu}}{N},
        \end{split}
    \end{equation}
    where $N_{\mu,1}^L$ is the lower bound estimator of $N_{\mu,1}$ with failure probability less than $\varepsilon$ and $Y^L$ is the lower bound estimator for $Y_1[1-h(e_1)]$ with failure probability less than $2\varepsilon$. 

    We first discuss about $N_{\mu, 1}^L$. The random variable $N_{\mu, 1}$ follows a Bernoulli distribution whose expectation is $N_{\mu}\mu e^{-\mu}$. Therefore we have:	
	\begin{equation}
		\begin{split}
			N_{\mu, 1}^L = \frac{\mathbb{E}[N_{\mu, 1}]}{1+\delta_N} = \frac{N_{\mu}\mu e^{-\mu}}{1+\delta_{N}},
		\end{split}
	\end{equation}
	with $\delta_N$ being the solution to:
    \begin{equation}
        e^{-\frac{\delta^2}{2+\delta}\mathbb{E}[N_{\mu,1}]} = \varepsilon,
    \end{equation}
     which is exactly the expression in the first line of Eq.~\eqref{Eq:delta_in_exp}.

    As for $Y^L$, the general form is derived in last subsection, our task here is to figure out expressions for $\Delta_k, Z_{k}, B_k$ and compute $\delta_1,\delta_2$. As discussed in~\ref{subsec:improved}, we choose $a, b$ such that $a = \log(1-e) + 1, b = \log(1-e) - \log{e}$, where:	
	\begin{equation}
		e = \frac{\mu^2 E_{\nu} Q_{\nu} e^{\nu} - \nu^2 E_{\mu} Q_{\mu} e^{\mu}}{\mu^2 Q_{\nu} e^{\nu} - \nu^2 Q_{\mu} e^{\mu}}.
	\end{equation}
	
	The form of $T_{k}$ is already derived in Eq.~\eqref{Eq:XandY}. Now we rewrite them into $\Delta_k \cdot (Z_k + B_k)$ form. One possible decomposition is as follows:	
	\begin{equation}
		\begin{split}
			T_{1} &= (a-b)\frac{\mu e^{\nu}}{N_{\nu}}(\mu-\nu) \cdot\sum_{i=0}^{\infty}\sum_{j=1}^{m_i} 1_{\chi_{i,j}=\nu}, \\
			T_{2} &= a\Bigl(\frac{\mu^2e^{\nu}}{N_{\nu}} + \frac{\nu^2e^{\mu}}{N_{\mu}}\Bigr)\cdot\sum_{i=0}^{\infty}\sum_{j=1}^{c_i} \left(1_{\chi_{i,j}=\nu} - \frac{\nu^2e^{\mu}/N_{\mu}}{\nu^2e^{\mu}/N_{\mu} + \mu^2e^{\nu}/N_{\nu}}\right).
		\end{split}
	\end{equation}
	
	In such a case, we have $\Delta_1 = (a-b)\frac{\mu e^{\nu}}{N_{\nu}}(\mu-\nu), B_1 = 0$ and $\Delta_2 = a\Bigl(\frac{\mu^2e^{\nu}}{N_{\nu}} + \frac{\nu^2e^{\mu}}{N_{\mu}}\Bigr), B_2 = -\frac{\nu^2e^{\mu}/N_{\mu}}{\nu^2e^{\mu}/N_{\mu} + \mu^2e^{\nu}/N_{\nu}}$. Meanwhile, $Z_{1,ij}$ and $Z_{2,ij}$ are both $1_{\chi_{i,j}=\nu}$.  Recalling~\ref{subsec:Chernoff}, the $\delta_1, \delta_2$ can be computed as:	
	\begin{equation}
		\begin{split}
			\delta_1 &= \frac{-3\ln{\varepsilon} + \sqrt{-8m_{\nu}\ln{\varepsilon} + \ln^2{\varepsilon}}}{2(m_{\nu} + \ln\varepsilon)}, \\
			\delta_2 &= \frac{-3\ln{\varepsilon} + \sqrt{-8c_{\nu}\ln{\varepsilon} + \ln^2{\varepsilon}}}{2(c_{\nu} + \ln\varepsilon)}.
		\end{split}
	\end{equation}
	
	Utilizing the relation that $\sum_{i=0}^{\infty}\sum_{j=1}^{c_i}1 = c_{\mu} + c_{\nu}$ and $\sum_{i=0}^{\infty}\sum_{j=1}^{m_i}1 = m_{\mu} + m_{\nu}$, we combine the parameters above and replace $\sum_{i,j}Z_{k,ij}$ with their sample in the experiment, i.e., $m_{\nu}$ and $c_{\nu}$, to obtain the final formula.

\end{proof}

Note that this result is practically useful only when the data size is not too small. When $m_{\nu}$ is too small, $\delta_1$ may exceed 1, making the result invalid. Details for handling this extreme case has already been discussed at the end of~\ref{subsec:Chernoff}.

\section{Simulation} \label{Sec:Simulation}
In this section, we present simulation results to justify our linear expansion and demonstrate the advantages of our method. Here, we adopt the parameter settings shown in Table~\ref{Table:parameters}. The experimental statistics are calculated by the following formulas:
\begin{equation}
	\begin{split}
		Q_{\mu} &= 1 - (1-Y_0) e^{-\mu \eta}, \\
		Q_{\nu} &= 1 - (1-Y_0) e^{-\nu \eta}, \\
		E_{\mu} Q_{\mu} &= e_0 Y_0 + e_d (1 - e^{-\mu \eta}), \\
		E_{\nu} Q_{\nu} &= e_0 Y_0 + e_d (1 - e^{-\nu \eta}),
	\end{split}
\end{equation}
where $\eta$ is the transmission loss $\eta = 10^{-\alpha l/10}$ and $l$ is the transmission fiber length. Note that parameters such as $e_d, \eta_d, Y_0$ are used solely to compute the statistics $Q_{\mu}, Q_{\nu}, E_{\mu}, E_{\nu}$ for simulation purposes. They are not assumed to be known or directly used in the security analysis. The parameter $\delta$ to bound the fluctuation is evaluated according to Lemma~\ref{lemma:analytical}. Following the observation made therein, we employ binary search to compute $\delta$ when $\phi < -100\ln\varepsilon$.

\begin{table}[htbp]
	\centering \caption{Simulation parameters. The signal and decoy state parameters are set as described in~\cite{Liao2017satelite}, and the fiber and detector parameters follow from some of the most recent experiments~\cite{Ma2018phase,Zhu_2024}. Here, $\mu$ and $\nu$ represent the intensities for the signal and decoy states, respectively, with $p_{\mu}$ and $p_{\nu}$ being their corresponding probabilities. $\eta_d$ denotes the detector efficiency and $f_e$ represents the error correction efficiency. The parameter $e_d$ refers to the average bit error rate for nonzero photon components and $Y_0$ is the dark count rate. $\alpha$ is the optical fiber loss and $\varepsilon$ is the failure probability.} \label{Table:parameters}
	\begin{tabular}{p{0.05\textwidth}p{0.05\textwidth}p{0.05\textwidth}p{0.05\textwidth}p{0.05\textwidth}p{0.05\textwidth}p{0.06\textwidth}p{0.08\textwidth}p{0.1\textwidth}p{0.08\textwidth}}
		\hline
		\hline		
		$\mu$ & $\nu$ & $p_{\mu}$ & $p_{\nu}$ & $\eta_d$ & $f_e$ & $e_d$ & $Y_0$ & $\alpha$ & $\varepsilon$ \\
		\hline 
		0.6 & 0.2 & 6/7 & 1/7 & 72$\%$ & 1.06 & 1.5$\%$ & $3\times 10^{-8}$ & 0.21dB/km & $10^{-10}$\\
		\hline
		\hline
	\end{tabular}
\end{table}

\subsection{Justification of Linear Expansion}\label{subsec:justification}
Here, we discuss the key assumptions underlying our linear expansion method. During the derivation, we relied on the following condition:
\begin{equation}\label{Eq:assumption}
	(a-b)\frac{\nu}{\mu+\nu}+(b-2a)\geq 0,
\end{equation}
where $a = 1 + \log(1 - e_t)$ and $b = \log(1 - e_t) - \log{e_t}$. It is evident that when $e_t$ approaches zero, this condition holds for any values of $\mu$ and $\nu$. However, choosing a very small $e_t$ can reduce the effectiveness of our method. The key question is: how large can $e_t$ be while still satisfying this condition?

Figure~\ref{fig:just} offers insights into this. In the extreme case where $\nu = \mu$, $e_t$ can reach values as high as 0.2. In more typical decoy-state setups, where $\nu/\mu$ is often less than $1/2$, $e_t$ can exceed 0.3, which corresponds to a relatively high error rate. Ideally, $e_t$ should be selected to be close to $e_1$. However, when $e_1$ is large, no secure keys can be generated, making this region irrelevant in practical settings. In scenarios where secure keys are feasible, our assumption consistently holds, allowing us to confidently use the linear bound derived in Eq.~\eqref{Eq:linearelax}.

\begin{figure}[htbp!]
	\includegraphics[width=0.35\linewidth]{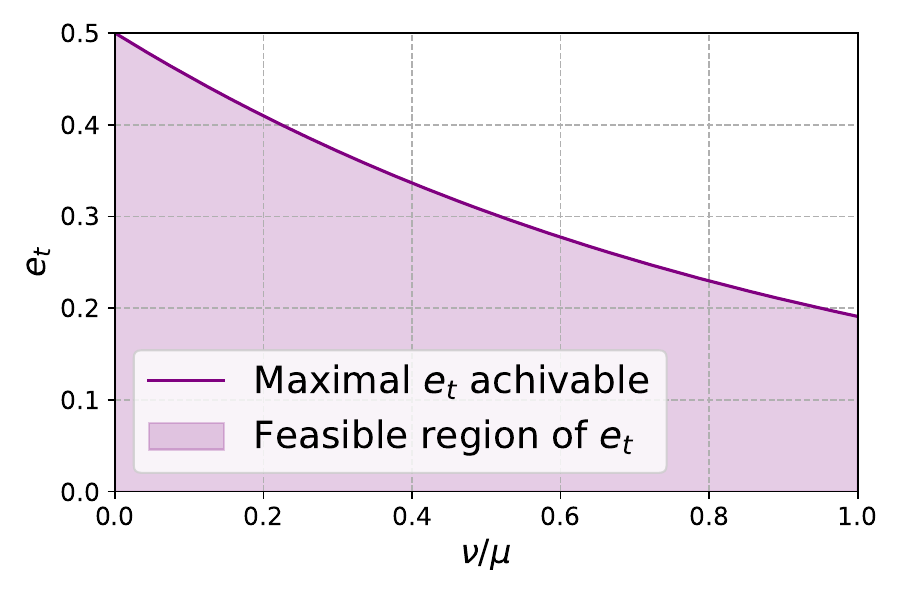}
	\caption{Feasibility of $e_t$ as a function of the intensity ratio between the signal state $\mu$ and the decoy state $\nu$. The feasible values of $e_t$ are those that do not violate the assumption stated in Eq.~\eqref{Eq:assumption}. This relationship helps to identify suitable values for $e_t$ in different experimental setups.} \label{fig:just}
\end{figure}

To be more specific, we recommend using the following expression for $e_t$:
\begin{equation}
	e_t = \frac{\mu^2 E_{\nu} Q_{\nu} e^{\nu} - \nu^2 E_{\mu} Q_{\mu} e^{\mu}}{\mu^2 Q_{\nu} e^{\nu} - \nu^2 Q_{\mu} e^{\mu}},
\end{equation}
as given in Eq.~\eqref{eq:starcase}. Figure~\ref{fig:error-count} illustrates the relationship between $e_t$ and the transmission distance. The figure shows that whenever a secure key can be generated, the condition in Eq.~\eqref{Eq:assumption} is satisfied, thereby justifying our use of the linear expansion.

\begin{figure}
	\centering 
	\includegraphics[width=0.35\linewidth]{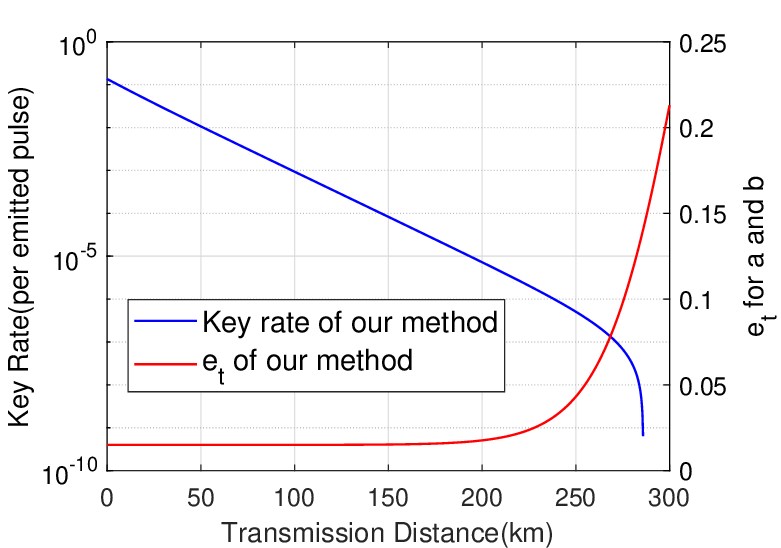}
	\caption{The relationship between $e_t$ and the key rate as a function of transmission distance when $\nu/\mu=1/3$. When secure keys are feasible, the corresponding $e_t$ is smaller than $0.15$, which would not violate the condition in Eq.~\eqref{Eq:assumption}.}
	\label{fig:error-count}
\end{figure}

Recall that in Theorem~\ref{thm:simple}, we bound $Y_1[1-h(e_1)]$ by $Y_1^*[1-h(e_1^*)] + \frac{\nu}{2}(1+\log{e_1^*})e_2^*Y_2^*$. Figure~\ref{fig:rate_thr} provides further numerical results to support our claim that the term $\frac{\nu}{2}(1+\log e_1^*)e_2^*Y_2^*$ is significantly smaller than $Y_1^*[1-h(e_1^*)]$ in most practical settings. 

\begin{figure}[htbp!]
	\includegraphics[width=0.35\linewidth]{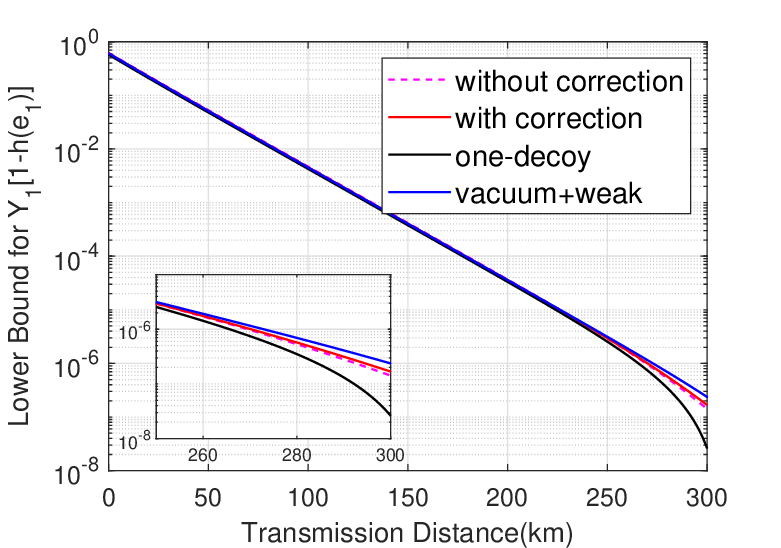}
	\caption{The magenta dashed line represents $Y_1^*[1-h(e_1^*)]$. The red line shows the value after applying the correction term. The black line corresponds to the standard one-decoy method, while the blue line shows the result from the vacuum+weak decoy method.} \label{fig:rate_thr}
\end{figure}

Notably, the correction term is negative at shorter transmission distances but becomes positive as the distance increases. This transition is reasonable because, at longer distances, $Y_0$ contributes substantially to both valid clicks and bit errors, making it impossible for $Y_0$ to be zero, given the values of $Q_{\mu}, Q_{\nu}, E_{\mu},$ and $E_{\nu}$. The initial derivation of $Y_1[1-h(e_1)]$ assumed $Y_0 = 0$, which is not physically plausible under these conditions, and the correction term effectively compensates for this assumption. Additionally, the positivity of this correction term provides an intuitive threshold for assessing the significance of $Y_0$. Specifically, when $\nu E_{\mu} Q_{\mu} e^{\mu} - \mu E_{\nu} Q_{\nu} e^{\nu} \leq 0$, $Y_0$ has minimal impact on the final key rate. However, as this determinant approaches positive values, the role of $Y_0$ becomes increasingly important.

\subsection{Key Rate Simulation}
Here, we evaluate the effectiveness of our enhanced statistical fluctuation analysis through a numerical comparison with existing Chernoff+Hoeffding methods. We first demonstrate the theoretical improvement provided by our linear relaxation approach and then evaluate the secret key generation rate and the maximum secure transmission distance achievable in a finite-key scenario. These comparative analyses highlight the practical benefits and improvements our method offers over traditional approaches in decoy-state QKD.

\begin{figure}[htbp]
	\centering
	\includegraphics[width=0.35\linewidth]{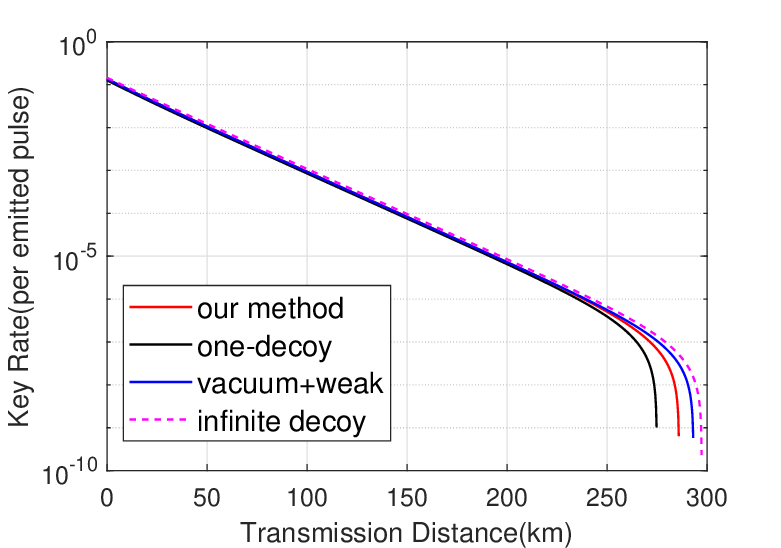}
	\caption{Comparison of key generation rates under fluctuation-free conditions.}
	\label{fig:thr_improve}
\end{figure}

We first illustrate the improvement in key generation rates under fluctuation-free conditions, as shown in Figure~\ref{fig:thr_improve}. The dashed magenta line represents the $Y_1[1-h(e_1)]$ result of the infinite decoy state case, calculated using:
\begin{equation}
	\begin{split}
		Y_1 &= Y_0 + (1 - Y_0)(1 - \eta), \\
		e_1 Y_1 &= e_0 Y_0 + e_d (1 - Y_0)(1 - \eta).
	\end{split}
\end{equation}
We observe that the vacuum + weak decoy state method theoretically achieves the highest key generation rate, while the one-decoy state method performs the worst due to insufficient information about $Y_0$. Our method, although not relying on vacuum states to estimate $Y_0$, shows comparable performance to the vacuum + weak decoy method across a considerable region.

\begin{figure}[htbp]
	\centering
	\includegraphics[width=0.35\linewidth]{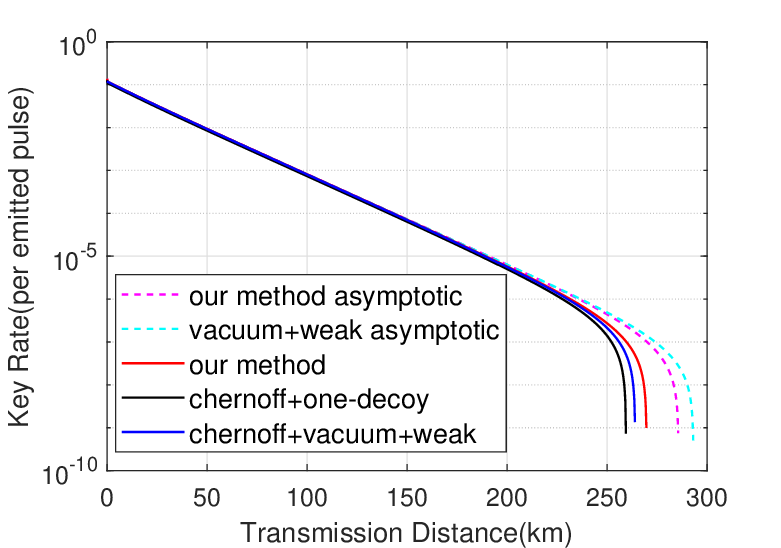}
	\caption{Comparison of key generation rates with a data size of $N = 10^{11}$.}
	\label{fig:fluctuation_rate}
\end{figure}

The situation changes significantly when statistical fluctuations are taken into account. Estimating $Y_0$ introduces considerable challenges for the vacuum + weak decoy state method when the data size is limited, supporting our choice of the one-decoy protocol. Figure~\ref{fig:fluctuation_rate} shows the key generation rate for $N = 10^{11}$. We use the original probability settings of $6:1:1$ for the signal state, weak decoy, and vacuum decoy states, as described in~\cite{Liao2017satelite}. Although the vacuum + weak decoy state method theoretically outperforms our method, this advantage diminishes under finite data conditions due to inaccuracies in estimating $Y_0$. Since $Y_0$ is on the order of $10^{-8}$, its fluctuation can become comparable to or even larger than theoretical bounds derived by the one-decoy method when $N$ is not large enough. At shorter transmission distances, relative fluctuations are smaller, allowing our method to achieve a higher key generation rate compared to the one-decoy method, consistent with the theoretical results, while the vacuum + weak decoy method also performs well. As the distance approaches the maximum transmission range, our fluctuation analysis demonstrates robustness over the other two methods. At $250$ km, our method's key generation rate is $2.31$ times that of the one-decoy method and $1.46$ times that of the vacuum + weak decoy method.

Finally, we compare the maximum transmission distance achievable with different data sizes. The results are shown in Figure~\ref{fig:distance}. For small $N$, the vacuum + weak decoy method effectively degenerates to the one-decoy method because the fluctuation in $Y_0$ becomes unreliably large, requiring us to use the bound from the one-decoy method. Both methods converge to their respective maximum distances when $N$ reaches approximately $10^{14}$. This is understandable since both use the Chernoff bound as the basis for fluctuation analysis. However, our method shows a better maximum distance at smaller data sizes. When $N = 10^{11}$, our method achieves a $10$ km improvement over the standard one-decoy method and a $6$ km improvement over the vacuum + weak decoy method.

\begin{figure}[htbp]
	\centering
	\includegraphics[width=0.35\linewidth]{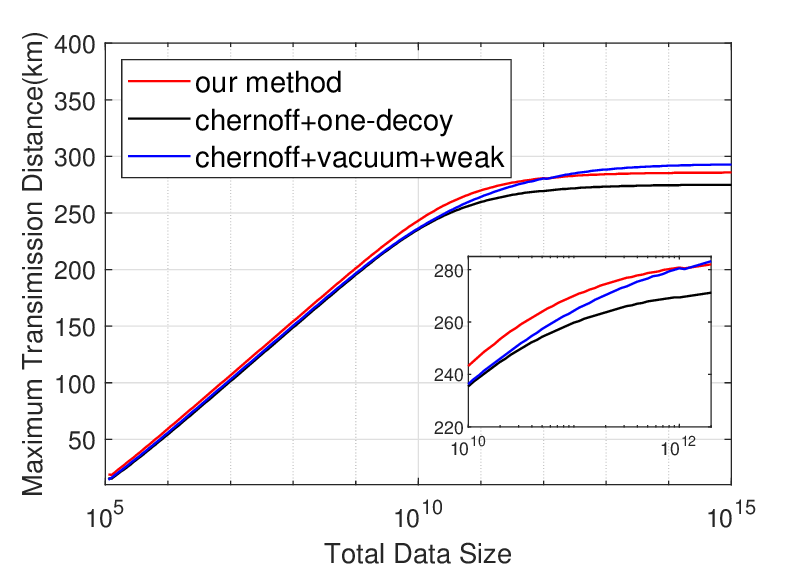}
	\caption{Comparison of maximum transmission distances.}
	\label{fig:distance}
\end{figure}

\section{Conclusion and Outlook} \label{sc:conclusion}
In this paper, we present an enhanced bound estimation technique and a refined fluctuation analysis to improve the finite-data performance of decoy-state QKD systems. Unlike previous approaches, we employ a linear expansion to simplify the Shannon entropy function and introduce a joint statistical fluctuation analysis framework to effectively handle the resulting linear relationships. The slope and intercept parameters of the linear expansion, similar to the state intensities and probabilities, can be optimized to achieve a higher key generation rate.

We focus primarily on the one-decoy state method due to its simplicity and experimental feasibility. However, our linear expansion and fluctuation analysis techniques are easily extendable to other decoy-state methods and QKD schemes, such as measurement-device-independent QKD or mode-pairing QKD. In this work, we assume a symmetric basis setting and do not account for different yields and error rates across various bases. Extending our analysis to cover asymmetric QKD protocols remains a direction for future exploration.

Furthermore, the fluctuation analysis framework introduced here can be applied to other quantum information post-processing tasks, provided the final outcome maintains a linear relationship with experimental statistics.

\begin{acknowledgements}
We extend our gratitude to Zhenyu Du and Yuwei Zhu for their valuable discussions and insights. This work was supported by the National Natural Science Foundation of China Grant No.~12174216 and the Innovation Program for Quantum Science and Technology Grants No.~2021ZD0300804 and No.~2021ZD0300702.
\end{acknowledgements}

\bibliography{bibDecoy.bib}
\end{document}